\documentclass[a4paper,11pt]{amsart}

\usepackage{amssymb,amscd,amsmath}

\theoremstyle{plain}

\newtheorem{thm}{Theorem}[section]
\newtheorem{prop}[thm]{Proposition}
\newtheorem{lem}[thm]{Lemma}

\newtheorem{obs}[thm]{Observation}

\theoremstyle{definition}

\theoremstyle{remark}

\newcommand{\forme}[1]{}

\title{On $k$-rainbow domination in middle graphs}

\author[Kim]{Kijung Kim}
\address{Department of Mathematics, Pusan National University, Busan 46241, Republic of Korea}
\email{knukkj@pusan.ac.kr}

\date{\today}
\begin{document}

\begin{abstract}

Let $G$ be a finite simple graph with vertex set $V(G)$ and edge set $E(G)$.
A function $f : V(G) \rightarrow \mathcal{P}(\{1, 2, \dotsc, k\})$ is a \textit{$k$-rainbow dominating function} on $G$ if
for each vertex $v \in V(G)$ for which $f(v)= \emptyset$, it holds that $\bigcup_{u \in N(v)}f(u) = \{1, 2, \dotsc, k\}$.
The weight of a $k$-rainbow dominating function is the value $\sum_{v \in V(G)}|f(v)|$.
The \textit{$k$-rainbow domination number} $\gamma_{rk}(G)$ is the minimum weight of a $k$-rainbow dominating function on $G$.
In this paper, we initiate the study of $k$-rainbow domination numbers in middle graphs.
We define the concept of a middle $k$-rainbow dominating function, obtain some bounds related to it
and determine the middle $3$-rainbow domination number of some classes of graphs.
We also provide upper and lower bounds for the middle $3$-rainbow domination number of trees in terms of the matching number.
In addition, we determine the $3$-rainbow domatic number for the middle graph of paths and cycles.

\bigskip

\end{abstract}

\maketitle

\insert\footins{\footnotesize
This research was supported by Basic Science Research Program through the National Research Foundation of Korea funded by the Ministry of Education (2020R1I1A1A01055403).}

\section{Introduction}\label{sec:intro}

Let $G=(V,E)$ be a connected undirected graph with the vertex set $V=V(G)$ and edge set $E=E(G)$.
The \textit{order} of $G$ is defined as the cardinality of $V$.
The \textit{open neighborhood} of $v \in V(G)$ is the set $N(v) = \{ u \in V(G) \mid uv \in E(G)\}$ and
the \textit{closed neighborhood} of $v \in V(G)$ is the set $N[v]:= N(v) \cup \{v\}$.
The \textit{degree} of $v \in V(G)$ is defined as the cardinality of $N(v)$, denoted by $deg_G(v)$.
When no confusion arises, we may delete the subscript $G$ in $deg_G(v)$.
The \textit{maximum degree} and \textit{minimum degree} of $G$ are denoted by $\Delta(G)$ and $\delta(G)$, respectively
We write $P_n$, $C_n$ and $K_n$ for a path, a cycle and a complete graph, respectively.

In \cite{HY}, Hamada and Yoshimura defined the middle graph of a graph.
The \textit{middle graph} $M(G)$ of a graph $G$ is the graph obtained by subdividing each edge of $G$ exactly once
and joining all these newly introduced vertices of adjacent edges of $G$.
The precise definition of $M(G)$ is as follows.
The vertex set $V(M(G))$ is $V(G) \cup E(G)$.
Two vertices $v, w \in V(M(G))$ are adjacent in $M(G)$ if
(i) $v, w \in E(G)$ and $v, w$ are adjacent in $G$ or
(ii) $v \in V(G)$, $w \in E(G)$ and $v, w$ are incident in $G$.

In the graph domination, a set of vertices is selected as guards such that each vertex not selected has a guard as a neighbor.
As a generalization of the graph domination, Bre$\acute{s}$ar et al. introduced the concept of rainbow domination in \cite{BHR}.
In the $k$-rainbow domination, $k$-different types of guards are required in the neighborhood of a non-selected vertex.
Let $[k]$ be the set of positive integers at most $k$.
A function $f : V(G) \rightarrow \mathcal{P}([k])$ is a \textit{$k$-rainbow dominating function} on $G$ if
for each vertex $v \in V(G)$ for which $f(v)= \emptyset$, it holds that $\bigcup_{u \in N(v)}f(u) = [k]$.
The weight of a $k$-rainbow dominating function is the value $\sum_{v \in V(G)}|f(v)|$.
The \textit{$k$-rainbow domination number} $\gamma_{rk}(G)$ is the minimum weight of a $k$-rainbow dominating function on $G$.
In \cite{CWZ}, Chang et al. proved that the $k$-rainbow domination is NP-complete.
So, it is worthwhile to determine the $k$-rainbow domination numbers of some classes of graphs.
Indeed, there are many papers on the $2$-rainbow domination.
The latest survey of the $2$-rainbow domination is introduced in \cite{B}.
For $k \geq 3$, it is more difficult to determine the $k$-rainbow domination number of a graph.
The following are a few results on the $3$-rainbow domination number.
In \cite{SLYXPZ}, Shao et al. determined the $3$-rainbow domination numbers of paths, cycles and generalized Petersen graphs $P(n,1)$.
In \cite{WWDAKL}, Wang et al. determined the $3$-rainbow domination number of $P_3 \square P_n$.
In \cite{GXY}, Gao et al. determined the $3$-rainbow domination numbers of $C_3 \square C_m$ and $C_4 \square C_m$.
In \cite{CK}, Cynthia et al. determined the $3$-rainbow domination number of circulant graph $G(n; \pm\{1,2,3\})$.
In \cite{FKY}, Furuya et al. proved that for every connected graph $G$ of order $n \geq 8$ with $\delta(G) \geq 2$,
$\gamma_{r3}(G) \leq \frac{5n}{6}$.

To study the $k$-rainbow domination number in the class of middle graphs, we define the following concept.
For $v \in V(G)$, we denote $\{ e \in E(G) \mid e ~\text{is incident with}~ v\}$ by $N_M(v)$.
For $e \in E(G)$, we denote $\{ x \in V(G) \cup E(G) \mid x ~\text{is either adjacent or incident with}~ e\}$ by $N_M(e)$.
We write $N_M[x] = N_M(x) \cup \{x\}$.
A \textit{middle $k$-rainbow dominating function} (MkRDF) on a graph $G$ is a function $f: V \cup E \rightarrow \mathcal{P}([k])$ such that
every element $x \in V \cup E$ for which $f(x)=\emptyset$ satisfies $\bigcup_{v \in N_M(x)}f(v) = [k]$.
A middle $k$-rainbow dominating function $f$ gives an ordered partition $(V_0 \cup E_0, V_1 \cup E_1, V_2 \cup E_2, \dotsc, V_k \cup E_k)$,
where $V_i := \{ x \in V \mid |f(x)|=i \}$ and  $E_i := \{ x \in  E \mid |f(x)|=i \}$.
The \textit{weight} of a middle $k$-rainbow dominating function $f$ is $\omega(f) :=\sum_{x \in V \cup E} |f(x)|$.
The \textit{middle $k$-rainbow domination number} $\gamma_{rk}^\star(G)$ of $G$ is the minimum weight of a middle $k$-rainbow dominating function of $G$.
A \textit{$\gamma_{rk}^\star(G)$-function} is a MkRDF on $G$ with weight $\gamma_{rk}^\star(G)$.
We remark that $\gamma_{rk}^\star(G) = \gamma_{rk}(M(G))$ for any graph $G$.
In \cite{Kim}, it was considered only $2$-rainbow domination numbers of the middle graphs.
In this paper, we initiate the study of the middle $k$-rainbow domination in graphs.
In particular, we determine the exact value of middle $3$-rainbow domination numbers of some classes of graphs.
A \textit{matching} in a graph $G$ is a set of pairwise nonadjacent edges.
The maximum number of edges in a matching of a graph $G$ is called the matching number of $G$ and denoted by $\alpha'(G)$.
We provide upper and lower bounds for the middle $3$-rainbow domination number of trees in terms of the matching number.
A set $\{f_1, \dotsc, f_d\}$ of $k$-rainbow dominating functions of $G$ is called a \textit{$k$-rainbow dominating family} on $G$
if $\sum_{i=1}^d |f_i(v)| \leq k$ for each $v \in V(G)$.
The maximum number of functions in a $k$-rainbow dominating family on $G$ is the \textit{$k$-rainbow domatic number} of $G$, denoted by
$d_{rk}(G)$.
It is known that $k$-rainbow domatic number is well-defined and $d_{rk}(G) \geq k$ for every graph $G$ (See \cite{SV}).
We determine the $3$-rainbow domatic number for the middle graph of paths and cycles.

The rest of this section, we present some necessary terminology and notation.
For terminology and notation on graph theory not given here, the reader is referred to \cite{BM}.
Let $T$ be a (rooted) tree.
A \textit{leaf} of $T$ is a vertex of degree one.
A \textit{pendant edge} is an edge incident with a leaf.
A \textit{support vertex} is a vertex adjacent to a leaf.
For a vertex $v$, $C(v)$ denote the set of the children of $v$.
$D[v]$ denote the set of the descendants and $v$.
The subtree induced by $D[v]$ is denoted by $T_v$.
We write $K_{1,n-1}$ for the \textit{star} of order $n \geq 3$.
The \textit{double star} $DS_{p,q}$, where $p, q \geq 1$, is the graph obtained
by joining the centers of two stars $K_{1,p}$ and $K_{1,q}$.
A \textit{healthy spider} $S_{t,t}$ is the graph from a star $K_{1,t}$ by subdividing each edges of $K_{1,t}$.
A \textit{wounded spider} $S_{t,r}$ is the graph from a star $K_{1,t}$ by subdividing $r$ edges of $K_{1,t}$,
where $r \leq t-1$.
Note that a star $K_{1,t}$ is a wounded spider $S_{t,0}$.
For a graph $G$ and its subset $S$, $G - S$ denotes the subgraph of $G$ induced by $V(G) \setminus V(S)$.
A \textit{diametral path} of $G$ is a path with the length which equals the diameter of $G$.
The \textit{complement} of $G=(V,E)$ is the graph $(V, \overline{E})$, which is denoted by $\overline{G}$,
where $uv \in \overline{E}$ if and only if $uv \not\in E$.

\section{General bounds of the middle $k$-rainbow domination number}\label{sec:basic}

In this section, we obtain general bounds of the middle $k$-rainbow domination number.
First, we begin by giving a simple lower bound on the middle $k$-rainbow domination number.

\begin{obs}\label{prop:lower-bound}
If $G$ is a graph with $|V(G)| + |E(G)| \geq k$,
then $\gamma_{rk}^\star(G) \geq k$.
\end{obs}
\begin{proof}
Let $f$ be a $\gamma_{rk}^\star(G)$-function.
If $f(x) = \emptyset$ for some $x \in V(G) \cup E(G)$, then clearly $\bigcup_{y \in N_M(x)}f(y) = \{ 1, \dotsc, k\}$.
If $f(x) \neq \emptyset$ for all $x \in V(G) \cup E(G)$, then it follows from $|V(G)| + |E(G)| \geq k$ that $\gamma_{rk}^\star(G) \geq k$.
\end{proof}

\begin{prop}\label{prop:3rain}
Let $G$ be a graph of order $n \geq 2$.
Then $\gamma_{r3}^\star(G) = 3$ if and only if $G \in \{\overline{K_3}, P_2\}$.
\end{prop}
\begin{proof}
If $G \in \{\overline{K_3}, P_2\}$, then clearly $\gamma_{r3}^\star(G) = 3$.
Conversely, assume that $\gamma_{r3}^\star(G) = 3$ and let $f$ be a $\gamma_{r3}^\star(G)$-function.
If there exists $x \in V(G) \cup E(G)$ such that $f(x)=[3]$, then $x \in E(G)$ for otherwise $x$ can not dominate the other vertices.
Thus, $G = P_2$.

Now assume that there is no element with weight $3$.
If $f(v)= \emptyset$ for some $v \in V(G)$, then there exist at least two edges $e_1, e_2$ incident to $v$ such that $f(e_1)$ and $f(e_2)$ are not empty. But, end vertices of $e_1, e_2$ except for $v$ are not dominated, a contradiction.
Thus, every vertex in $V(G)$ has non-zero weight so that there are at most three vertices in the graph $G$.
One can easily check that $G = \overline{K_3}$ or $P_2$.
\end{proof}

%\begin{thm}\label{thm:bound-upper-lower}
%If $G$ is a connected graph, then
%$\text{max}\{ \gamma_{rk}(G), \gamma_{rk}'(G) \} \leq \gamma_{rk}^\star(G) \leq |V(G)| + \gamma_{rk}'(G)$.
%\end{thm}
%\begin{proof}
%\end{proof}

\begin{thm}\label{thm:del-vertex}
If $G$ is a graph and $v \in V(G)$, then
$\gamma_{rk}^\star(G) - \text{min}\{\Delta(G) +1, k\} \leq \gamma_{rk}^\star(G - v) \leq \gamma_{rk}^\star(G)$.
\end{thm}
\begin{proof}
First, we claim that $\gamma_{rk}^\star(G) - \text{min}\{\Delta(G) +1, k\} \leq \gamma_{rk}^\star(G - v)$.
Let $f$ be a $\gamma_{rk}^\star(G - v)$-function.
If $k \leq \Delta(G) +1$, then define
$g : V(G) \cup E(G) \rightarrow \mathcal{P}([k])$ by $g(v)=[k]$, $g(x)=\emptyset$ for $x \in N_M(v)$ and $g(x)=f(x)$ otherwise.
If $k > \Delta(G) +1$, then define
$g : V(G) \cup E(G) \rightarrow \mathcal{P}([k])$ by $g(x) =\{1\}$ for each $x \in N_M[v]$ and $g(x)=f(x)$ otherwise.
Clearly, $g$ is a MkRDF of $G$ with weight at most $\gamma_{rk}^\star(G - v)$ + min$\{\Delta(G) +1, k\}$.
Thus, $\gamma_{rk}^\star(G) - \text{min}\{\Delta(G) +1, k\} \leq \gamma_{rk}^\star(G - v)$.

Next, we claim that $\gamma_{rk}^\star(G - v) \leq \gamma_{rk}^\star(G)$.
Let $f$ be a $\gamma_{rk}^\star(G)$-function.
Define $h : V(G-v) \cup E(G-v) \rightarrow \mathcal{P}([k])$ by $h(u) = f(u) \cup f(uv)$ for $u \in N(v)$ and $h(x) =f(x)$ otherwise.
Then clearly $h$ is a MkRDF of $G - v$ with weight $\gamma_{rk}^\star(G)$.
\end{proof}

\begin{thm}\label{thm:add-edge12}
Let $G$ be a graph.
Then
\begin{enumerate}
\item $\gamma_{rk}^\star(G) -k  \leq \gamma_{rk}^\star(G + e) \leq \gamma_{rk}^\star(G) +1$ for $e \in E(\overline{G})$,
\item $\gamma_{rk}^\star(G) -1  \leq \gamma_{rk}^\star(G - e) \leq \gamma_{rk}^\star(G) +k$ for $e \in E(G)$.
\end{enumerate}
\end{thm}
\begin{proof}
(i) First, we show that $\gamma_{rk}^\star(G + e) \leq \gamma_{rk}^\star(G) +1$.
Let $f$ be $\gamma_{rk}^\star(G)$-function.
Clearly, we can extend $f$ to a MkRDF of $G +e$ by assigning $\{1\}$ to $f(e)$.

Next, we claim that $\gamma_{rk}^\star(G) -k  \leq \gamma_{rk}^\star(G + e)$.
Let $f$ be $\gamma_{rk}^\star(G +e)$-function and let $e=uv$.
If $f(e) =\emptyset$, then clearly the function $f|_{V(G)\cup E(G)}$ is a MkRDF of $G$.
This implies $\gamma_{rk}^\star(G) -k  < \gamma_{rk}^\star(G) \leq \gamma_{rk}^\star(G + e)$.

Now assume that $f(e) \neq \emptyset$.
Then define $g : V(G) \cup E(G) \rightarrow \mathcal{P}([k])$ by
$g(u) =f(u) \cup f(e)$, $g(v) = f(v) \cup f(e)$ and $g(x) =f(x)$ otherwise.
Clearly, $g$ ia a MkRDF of $G$ with weight $\gamma_{rk}^\star(G + e) + |f(e)|$.
Thus, $\gamma_{rk}^\star(G) -k \leq \gamma_{rk}^\star(G) -|f(e)| \leq \gamma_{rk}^\star(G + e)$.

\vskip5pt
(ii) By (i), $\gamma_{rk}^\star(G -e) -k \leq \gamma_{rk}^\star((G - e) +e) \leq \gamma_{rk}^\star(G -e) +1$.
This implies (ii).
\end{proof}

%\begin{lem}\label{lem:spiderm}
%For a spider $S_{t,r}$ $(r \geq 1)$,
%$\gamma_{r3}^\star(S_{t,r}) \leq \frac{3(1+t)}{2}$.
%\end{lem}
%\begin{proof}
%The order of $S_{t,r}$ is $1+r+t$.
%\end{proof}

\begin{thm}\label{thm:tree-order}
Let $T$ be a tree of order $n\geq 3$.
Then $\gamma_{r3}^\star(T) \leq \frac{3n}{2}$.
\end{thm}
\begin{proof}
We proceed by induction on the order $n$ of $T$.
Obviously, the statement is true for a path $P_3$.

Let $T$ be a tree of order $n \geq 4$.
Suppose that every tree $T'$ of order $n' (< n)$ satisfies $\gamma_{r3}^\star(T') \leq \frac{3n'}{2}$.
Let $f$ be a $\gamma_{r3}^\star(T')$-function.
If $T$ is a star $K_{1,n-1}$, then $\gamma_{r3}^\star(T)=n+1 \leq \frac{3n}{2}$.
Assume that $T$ is a double star $DS_{p,q}$ with $p \geq q \geq 1$.
Then $p+q+3 = n+1 =\gamma_{r3}^\star(T) \leq \frac{3n}{2}$.
Now we assume that $T$ is neither a star or a double star.
Then it is easy to see that $T$ has diameter at least four.
Among all of diametrical paths in $T$, we choose $x_0x_1\dotsc x_d$ so that it maximizes the degree of $x_{d-1}$.
Root $T$ at $x_0$.
We divide our consideration into three cases.

\vskip5pt
\textbf{Case 1.} $deg_T(x_{d-1}) = t \geq 3$.

Now $T_{x_{d-1}} \cong K_{1, t-1}$.
Let $T' = T - T_{x_{d-1}}$.
Applying the induction hypothesis to $T'$, we have $\gamma_{r3}^\star(T') \leq \frac{3(n-t)}{2}$.
Define $g : V(T) \cup E(T) \rightarrow \mathcal{P}([3])$ by $g(x_{d-1})=g(x_{d-2}x_{d-1})=g(x_d)=\emptyset$,
$g(x_{d-1}x_d)=[3]$, $g(x_{d-1}x)=\emptyset$ and $g(x)=\{1\}$ for $x \in C(x_{d-1}) \setminus \{x_d\}$ and
$g(x)=f(x)$ otherwise.
Then clearly $g$ is a M3RDF of $T$ and so
\[\gamma_{r3}^\star(T) \leq \omega(g)= \gamma_{r3}^\star(T') + t+ 1 \leq \frac{3(n-t)}{2} + t+ 1 = \frac{3n-t+2}{2} < \frac{3n}{2}.\]

\vskip5pt
\textbf{Case 2.} $deg_T(x_{d-1}) = 2$ and $deg_T(x_{d-2}) = t \geq 3$.

Now $T_{x_{d-2}} \cong S_{t-1, r}$, where $1 \leq r \leq t-1$.
Let $T' = T - T_{x_{d-2}}$.
Applying the induction hypothesis to $T'$, we have $\gamma_{r3}^\star(T') \leq \frac{3(n-t-r)}{2}$.
Define $g : V(T) \cup E(T) \rightarrow \mathcal{P}([3])$ by $g(x_{d-3}x_{d-2})=g(x_{d-2})= g(x_{d-1})=g(x_{d-1}x_d)=\emptyset$,
$g(x_{d-2}x_{d-1})=[3]$,
$g(x_{d-2}x)=\emptyset$ for $x \in C(x_{d-2}) \setminus \{x_{d-1}\}$,
$g(x)=\{1\}$ if $x \in C(x_{d-2}) \setminus \{x_{d-1}\}$ is a leaf,
$g(x)=\emptyset$, $g(xy)=[3]$ and $g(y)=\emptyset$ if $x \in C(x_{d-2}) \setminus \{x_{d-1}\}$ is a support vertex, where
$y \in C(x)$,
and $g(x)=f(x)$ otherwise.
Then clearly $g$ is a M3RDF of $T$ and so
\[\gamma_{r3}^\star(T) \leq \omega(g)= \gamma_{r3}^\star(T') + t+ 2r \leq \frac{3(n-t-r)}{2} + t+ 2r = \frac{3n-t+r}{2} < \frac{3n}{2}.\]

\vskip5pt
\textbf{Case 3.} $deg_T(x_{d-1}) = 2$ and $deg_T(x_{d-2}) =2$.

Now $T_{x_{d-2}} \cong P_3$.
Let $T' = T - T_{x_{d-2}}$.
Applying the induction hypothesis to $T'$, we have $\gamma_{r3}^\star(T') \leq \frac{3(n-3)}{2}$.
Define $g : V(T) \cup E(T) \rightarrow \mathcal{P}([3])$ by $g(x_{d-3}x_{d-2})= g(x_{d-2}) =g(x_{d-1}) =g(x_{d-1}x_d)=\emptyset$,
$g(x_{d-2}x_{d-1})=[3]$, $g(x_d)=\{1\}$ and $g(x)=f(x)$ otherwise.
Then clearly $g$ is a M3RDF of $T$ and so
\[\gamma_{r3}^\star(T) \leq \omega(g)= \gamma_{r3}^\star(T') + 4 \leq \frac{3(n-3)}{2} + 4 = \frac{3n-1}{2} < \frac{3n}{2}.\]
\end{proof}

\section{The middle $3$-rainbow domination number of paths, cycles and complete graphs}\label{sec:main1}

In this section, we determine the middle $3$-rainbow domination number of paths, cycles and complete graphs.

\begin{prop}\label{prop:path}
For $n \geq 2$,
\begin{equation*}
\gamma_{r3}^\star(P_n) = \left\{
                      \begin{array}{ll}
                     \frac{4n-1}{3} & \hbox{for $n \equiv 1$ (mod $3$);} \\
                     \frac{4n+1}{3} & \hbox{for $n \equiv 2$ (mod $3$);} \\
                     \frac{4n}{3} & \hbox{for $n \equiv 0$ (mod $3$).}
                      \end{array}
                     \right.
\end{equation*}
\end{prop}
\begin{proof}
One can check that $\gamma_{r3}^\star(P_2)=3$, $\gamma_{r3}^\star(P_3)=4$, $\gamma_{r3}^\star(P_4)=5$,
$\gamma_{r3}^\star(P_5)= 7$ and $\gamma_{r3}^\star(P_6)= 8$.
We proceed by induction on $n$.
Assume that $P_n = v_1v_2 \dotsc v_n$ and let $x_{2i-1} =v_i$ for $1 \leq i \leq n$ and $x_{2i} =v_iv_{i+1}$ for $1 \leq i \leq n-1$.
Let $V(P_n) \cup E(P_n) = \{x_1, \dotsc, x_{2n-1}\}$ and $f$ be a $\gamma_{r3}^\star(P_n)$-function.
We divide our consideration into three cases.

\vskip5pt
\textbf{Case 1.} $n \equiv 1$ (mod $3$).

First, we claim that $\gamma_{r3}^\star(P_n) \geq \frac{4n-1}{3}$.
Assume that $n \geq 7$.
It is easy to see that $\sum_{i=1}^5|f(x_{i})| \geq 3$.
If $\sum_{i=1}^5|f(x_{i})| = 3$, then $|f(x_{4})| + |f(x_{5})|=0$.
To dominate $x_5$, it must be $f(x_6) =\{1,2,3\}$.
Thus, it is easy to see that $\sum_{i=1}^{7}|f(x_{i})| \geq 6$.
Define $h : V(P_n - \{v_1, \dotsc, v_4\}) \cup E(P_n - \{v_1, \dotsc, v_4\}) \rightarrow \mathcal{P}([3])$ by $h(x_{9}) = f(x_{8}) \cup f(x_{9})$ and $h(x_i) =f(x_i)$ for $10 \leq i \leq 2n-1$.
Clearly, $h$ is a M3KDF of $P_{n-4}$ with weight at most $\omega(f)-6$.
By the induction hypothesis, we have
\[\gamma_{r3}^\star(P_n) \geq \omega(h) +6 \geq  \gamma_{r3}^\star(P_{n-4}) +6 = \frac{4(n-4)}{3} +6 = \frac{4n+2}{3}.\]

Now assume that $\sum_{i=1}^5|f(x_{i})| \geq 4$.
Define $g : V(P_n - \{v_1, v_2, v_3\}) \cup E(P_n - \{v_1, v_2, v_3\}) \rightarrow \mathcal{P}([3])$ by $g(x_7) = f(x_6) \cup f(x_7)$
and $g(x_i) =f(x_i)$ for $8 \leq i \leq 2n-1$.
Clearly, $g$ is a M3KDF of $P_{n-3}$ with weight at most $\omega(f)-4$.
By the induction hypothesis, we have
\[\gamma_{r3}^\star(P_n) \geq \omega(g) +4 \geq \gamma_{r3}^\star(P_{n-3}) +4 =\frac{4(n-3)-1}{3} +4 = \frac{4n-1}{3}.\]

Next, we claim that $\gamma_{r3}^\star(P_n) \leq \frac{4n-1}{3}$.
Define $h : V(P_n) \cup E(P_n) \rightarrow \mathcal{P}([3])$ by $h(v_{1+3i})= \{1\}$ for $0 \leq i \leq \frac{n-1}{3}$,
$h(x_{4+6i})=[3]$ for $0 \leq i \leq \frac{n-4}{3}$ and $h(x)= \emptyset$ otherwise.
It is easy to see that $h$ is a M3RDF of $P_n$ with weight $\frac{4n-1}{3}$.
Thus, we have $\gamma_{r3}^\star(P_n) = \frac{4n-1}{3}$.

\vskip5pt
\textbf{Case 2.} $n \equiv 2$ (mod $3$).

By the same argument as in Case 1, we can show that $\gamma_{r3}^\star(P_n) \geq \frac{4n+1}{3}$.
Define $g : V(P_n) \cup E(P_n) \rightarrow \mathcal{P}([3])$ by $g(v_{1+3i})= \{1\}$ for $0 \leq i \leq \frac{n-2}{3}$,
$g(x_{4+6i})=[3]$ for $0 \leq i \leq \frac{n-5}{3}$, $g(v_n)=\{2,3\}$ and $g(x)= \emptyset$ otherwise.
It is easy to see that $g$ is a M3RDF of $P_n$ with weight $\frac{4n+1}{3}$.
Thus, we have $\gamma_{r3}^\star(P_n) = \frac{4n+1}{3}$.

\vskip5pt
\textbf{Case 3.} $n \equiv 0$ (mod $3$).

By the same argument as in Case 1, we can show that $\gamma_{r3}^\star(P_n) \geq \frac{4n}{3}$.
Define $h : V(P_n) \cup E(P_n) \rightarrow \mathcal{P}([3])$ by $h(v_{1+3i})= \{1\}$ for $0 \leq i \leq \frac{n-3}{3}$,
$h(x_{4+6i})=[3]$ for $0 \leq i \leq \frac{n-3}{3}$ and $h(x)= \emptyset$ otherwise.
It is easy to see that $h$ is a M3RDF of $P_n$ with weight $\frac{4n}{3}$.
Thus, we have $\gamma_{r3}^\star(P_n) = \frac{4n}{3}$.
\end{proof}

\begin{prop}\label{prop:cycle}
For $n \geq 3$,
\begin{equation*}
\gamma_{r3}^\star(C_n) = \left\{
                      \begin{array}{ll}
                     \frac{4n+2}{3} & \hbox{for $n \equiv 1$ (mod $3$);} \\
                     \frac{4n+1}{3} & \hbox{for $n \equiv 2$ (mod $3$);} \\
                     \frac{4n}{3} & \hbox{for $n \equiv 0$ (mod $3$).}
                      \end{array}
                     \right.
\end{equation*}
\end{prop}
\begin{proof}
Assume that $C_n = v_1v_2 \dotsc v_nv_1$,
where the subscript $k$ of $v_k$ is read by modulo $n$.
Let $f$ be a $\gamma_{r3}^\star(P_n)$-function such that
the size of $N:= \{ v_i \mid f(v_i) \neq \emptyset\}$ is as small as possible.
We divide our consideration into three cases.

\vskip5pt
\textbf{Case 1.} $n \equiv 1$ (mod $3$).

If there exists some $k$ such that $|f(v_{k})| + |f(v_{k}v_{k+1})| + |f(v_{k+1})| \geq 3$, then
define $g : V(P_n - \{v_k, v_{k+1}\}) \cup E(P_n - \{v_k, v_{k+1}\}) \rightarrow \mathcal{P}([3])$ by
$g(v_{k-1})= f(v_{k-1}) \cup f(v_{k-1}v_k)$, $g(v_{k+2}) = f(v_{k+1}v_{k+2}) \cup f(v_{k+2})$ and
$g(x)=f(x)$ otherwise.
It is easy to see that $g$ is a M3RDF of $P_{n-2}$ with weight at most $\gamma_{r3}^\star(C_n) -3$.
By Proposition \ref{prop:path}, we have
\[\gamma_{r3}^\star(C_n) \geq \omega(g) + 3 \geq \gamma_{r3}^\star(P_{n-2}) +3 = \frac{4(n-2)+1}{3} +3 = \frac{4n+2}{3}.\]

%Suppose that $|f(v_{i})| + |f(v_{i}v_{i+1})| + |f(v_{i+1})| \leq 2$ for each $i \in [n]$.

%If $|f(v_{k})| + |f(v_{k}v_{k+1})| + |f(v_{k+1})|=0$ for some $k \in [n]$, then $f(v_{k-1}v_k)$ or $f(v_{k+1}v_{k+2})$ are $[3]$,
%since $v_{k-1}v_k$ or $v_{k+1}v_{k+2}$ dominate $v_{k}v_{k+1}$.
%So, there exists some $i \in \{k-1, k+1\}$ such that $|f(v_{i})| + |f(v_{i}v_{i+1})| + |f(v_{i+1})| \geq 3$.

Assume that
\begin{equation}\label{A1}
|f(v_{k})| + |f(v_{k}v_{k+1})| + |f(v_{k+1})| \leq 2
\end{equation}
for each $k \in [n]$.

If there exists some $k \in [n]$ such that $|f(v_k)| \geq 2$,
then define $g : V(P_n - v_k) \cup E(P_n - v_k) \rightarrow \mathcal{P}([3])$ by $g(v_{k-1}) = f(v_{k-1}) \cup  f(v_{k-1}v_k)$,
$g(v_{k+1}) = f(v_{k+1}) \cup  f(v_kv_{k+1})$ and $g(x)=f(x)$ otherwise.
Clearly, $g$ is a M3RDF of $P_{n-1}$ with weight at most $\gamma_{r3}^\star(C_n) -2$.
By Proposition \ref{prop:path}, we have
\[\gamma_{r3}^\star(C_n) \geq \omega(g) + 2 \geq \gamma_{r3}^\star(P_{n-1}) +2 \geq \frac{4(n-1)}{3} +2 = \frac{4n+2}{3}.\]

Assume that $|f(v_k)| \leq 1$ for each $k \in [n]$.
If $|f(v_k)| = 0$ for each $k \in [n]$, then to dominate $v_k$, we must have $f(v_{k-1}v_k) \cup f(v_{k}v_{k+1}) =[3]$.
Thus, it follows from $n \geq 4$ that
\[\gamma_{r3}^\star(C_n) = \frac{1}{2}\sum_{1 \leq k \leq n} \sum_{x \in N_M(v_k)}|f(x)| = \frac{3n}{2}  \geq \frac{4n+2}{3}.\]

Now assume that $N$ is not empty.
For a fixed $v_i \in N$, if $|f(v_{i})| + |f(v_{i}v_{i+1})| + |f(v_{i+1})|=1$, then
$v_{i+1}v_{i+2}$ must dominate $v_{i+1}$ so that $f(v_{i+1}v_{i+2})=[3]$, a contradiction to (\ref{A1}).
Thus, $|f(v_{i})| + |f(v_{i}v_{i+1})| + |f(v_{i+1})|=2$.
By the same argument, we have $|f(v_{i})| + |f(v_{i-1}v_i)| + |f(v_{i-1})|=2$.
Suppose that $f(v_{i-1})\neq \emptyset$ and $f(v_{i+1})\neq \emptyset$.
Without loss of generality, assume that $f(v_i)= \{1\}$, $f(v_{l-1})=\{3\}$ and $f(v_{l+1})=\{2\}$.
Then to dominate $v_{i-1}v_i$ and $v_iv_{i+1}$, we must have $f(v_{l-2}v_{l-1})=\{2\}$ and $f(v_{l+1}v_{l+2})=\{3\}$.
It follows from (\ref{A1}) that $f(v_{l-2})= f(v_{l+2}) =\emptyset$.
To dominate $v_{i-2}$ and $v_{i+2}$, we must have $f(v_{l-3}v_{l-2})=\{1, 3\}$ and $f(v_{l+2}v_{l+3})=\{1, 2\}$.
It follows from (\ref{A1}) that $f(v_{l-3})= f(v_{l+3}) =\emptyset$.
To dominate $v_{i-3}$ and $v_{i+3}$, we must have $\{2\} \subseteq f(v_{l-4}v_{l-3})$ and $\{3\} \subseteq f(v_{l+3}v_{l+4})$.
Define $g : V(C_n) \cup E(C_n) \rightarrow \mathcal{P}([3])$
by $g(v_{i-3})= g(v_{i+3}) =\{1\}$, $g(v_{i-2}v_{i-1})= g(v_{i+1}v_{i+2}) =[3]$,
$g(v_{i-3}v_{i-2}) = g(v_{i-2}) = g(v_{i-1}) = g(v_{i-1}v_i)=  g(v_iv_{i+1}) = g(v_{i+1}) = g(v_{i+2}) = g(v_{i+2}v_{i+3}) =\emptyset$
and $g(x)=f(x)$ otherwise.
Then $g$ is a M3RDF of $C_n$ with weight $\omega(f)$.
As above, this implies that $\gamma_{r3}^\star(C_n) \geq \frac{4n+2}{3}$.

Now assume that for each $v_i \in N$, $|f(v_{i-1}v_i)|=1$ or $|f(v_iv_{i+1})|=1$.
Since $\sum_{x \in N_M[v_i]}|f(x)| \geq 2$ for $v_i \in N$ and $\sum_{x \in N_M[v_i]}|f(x)| =3$ for $v_i \in V(C_n) \setminus N$,
we have
\[\gamma_{r3}^\star(C_n) \geq \frac{3(n-t)}{2} + (t + \frac{t}{2}) = \frac{3n}{2} \geq \frac{4n+2}{3},\]
where $t=|N|$.

By Theorem \ref{thm:add-edge12}, $\gamma_{r3}^\star(C_n) \leq \gamma_{r3}^\star(P_n) +1$.
Thus, we have $\gamma_{r3}^\star(C_n) = \frac{4n+2}{3}$.

\vskip5pt
\textbf{Case 2.} $n \equiv 2$ (mod $3$).

If there exists some $k$ such that $|f(v_{k})| + |f(v_{k}v_{k+1})| + |f(v_{k+1})| \geq 3$, then
define $g : V(P_n - \{v_k, v_{k+1}\}) \cup E(P_n - \{v_k, v_{k+1}\}) \rightarrow \mathcal{P}([3])$ by
$g(v_{k-1})= f(v_{k-1}) \cup f(v_{k-1}v_k)$, $g(v_{k+2}) = f(v_{k+1}v_{k+2}) \cup f(v_{k+2})$ and
$g(x)=f(x)$ otherwise.
It is easy to see that $g$ is a M3RDF of $P_{n-2}$ with weight at most $\gamma_{r3}^\star(C_n) -3$.
By Proposition \ref{prop:path}, we have
\[\gamma_{r3}^\star(C_n) \geq \omega(g) + 3 \geq \gamma_{r3}^\star(P_{n-2}) +3 = \frac{4(n-2)}{3} +3 = \frac{4n+1}{3}.\]

%Suppose that $|f(v_{i})| + |f(v_{i}v_{i+1})| + |f(v_{i+1})| \leq 2$ for each $i \in [n]$.

%If $|f(v_{k})| + |f(v_{k}v_{k+1})| + |f(v_{k+1})|=0$ for some $k \in [n]$, then $f(v_{k-1}v_k)$ or $f(v_{k+1}v_{k+2})$ are $[3]$,
%since $v_{k-1}v_k$ or $v_{k+1}v_{k+2}$ dominate $v_{k}v_{k+1}$.
%So, there exists some $i \in \{k-1, k+1\}$ such that $|f(v_{i})| + |f(v_{i}v_{i+1})| + |f(v_{i+1})| \geq 3$.

Now suppose that
\[|f(v_{k})| + |f(v_{k}v_{k+1})| + |f(v_{k+1})| \leq 2\]
for each $k \in [n]$.

If there exists some $k \in [n]$ such that $|f(v_k)| \geq 2$,
then define $g : V(P_n - v_k) \cup E(P_n - v_k) \rightarrow \mathcal{P}([3])$ by $g(v_{k-1}) = f(v_{k-1}) \cup  f(v_{k-1}v_k)$,
$g(v_{k+1}) = f(v_{k+1}) \cup  f(v_kv_{k+1})$ and $g(x)=f(x)$ otherwise.
Clearly, $g$ is a M3RDF of $P_{n-1}$ with weight at most $\gamma_{r3}^\star(C_n) -2$.
By Proposition \ref{prop:path}, we have
\[\gamma_{r3}^\star(C_n) \geq \omega(g) + 2 \geq \gamma_{r3}^\star(P_{n-1}) +2 \geq \frac{4(n-1)-1}{3} +2 = \frac{4n+1}{3}.\]

Assume that $|f(v_k)| \leq 1$ for each $k \in [n]$.
By the same argument as Case 1, we have $\gamma_{r3}^\star(C_n) \geq \frac{4n+1}{3}$.

Define $h : V(C_n) \cup E(C_n) \rightarrow \mathcal{P}([3])$ by $h(v_{1+3i})= \{1\}$ for $0 \leq i \leq \frac{n-6}{3}$,
$h(v_{2+3i}v_{3+3i})=[3]$ for $0 \leq i \leq \frac{n-6}{3}$, $h(v_{n-1}v_n)=[3]$ and $h(x)= \emptyset$ otherwise.
It is easy to see that $h$ is a M3RDF of $C_n$ with weight $\frac{4n+1}{3}$.
Thus, we have $\gamma_{r3}^\star(C_n) = \frac{4n+1}{3}$.

\vskip5pt
\textbf{Case 3.} $n \equiv 0$ (mod $3$).

If there exists some $k \in [n]$ such that $|f(v_k)| \geq 1$,
then define $g : V(P_n - v_k) \cup E(P_n - v_k) \rightarrow \mathcal{P}([3])$ by $g(v_{k-1}) = f(v_{k-1}) \cup  f(v_{k-1}v_k)$,
$g(v_{k+1}) = f(v_{k+1}) \cup  f(v_kv_{k+1})$ and $g(x)=f(x)$ otherwise.
Clearly, $g$ is a M3RDF of $P_{n-1}$ with weight at most $\gamma_{r3}^\star(C_n) -1$.
By Proposition \ref{prop:path}, we have
\[\gamma_{r3}^\star(C_n) \geq \omega(g) + 1 \geq \gamma_{r3}^\star(P_{n-1}) +1 \geq \frac{4(n-1)+1}{3} +1 = \frac{4n}{3}.\]

Assume that $|f(v_i)| =0$ for each $i \in [n]$.
If there exists some $k \in [n]$ such that $|f(v_kv_{k+1})|=3$,
then define $g : V(P_n - \{v_k,v_{k+1}\}) \cup E(P_n - \{v_k, v_{k+1}\}) \rightarrow \mathcal{P}([3])$ by $g(v_{k-1}) = f(v_{k-1}) \cup  f(v_{k-1}v_k)$,
$g(v_{k+2}) = f(v_{k+2}) \cup  f(v_{k+1}v_{k+2})$ and $g(x)=f(x)$ otherwise.
Clearly, $g$ is a M3RDF of $P_{n-2}$ with weight at most $\gamma_{r3}^\star(C_n) -3$.
By Proposition \ref{prop:path}, we have
\[\gamma_{r3}^\star(C_n) \geq \omega(g) + 3 \geq \gamma_{r3}^\star(P_{n-2}) +3 \geq \frac{4(n-2)-1}{3} +3 = \frac{4n}{3}.\]

Now assume that  $|f(v_iv_{i+1})| \leq 2$ for each $i \in [n]$.
Then the assignment of edges under $f$ should start one of the following :
(i) $\{1\}, \{2,3\}, \{1\}, \{2,3\} \dotsc$, (ii) $\{2\}, \{1,3\}, \{2\}, \{1,3\} \dotsc$,
(iii)  $\{3\}, \{1,2\}, \{3\}, \{1,2\} \dotsc$.
Thus, one can check that
\[\gamma_{r3}^\star(C_n) \geq \lceil \frac{n}{2} \rceil + 2\lfloor \frac{n}{2} \rfloor \geq \frac{4n}{3}.\]
Therefore, we have $\gamma_{r3}^\star(C_n) \geq \frac{4n}{3}$.

Define $h : V(C_n) \cup E(C_n) \rightarrow \mathcal{P}([3])$ by $h(v_{1+3i})= \{1\}$ for $0 \leq i \leq \frac{n-3}{3}$,
$h(v_{2+3i}v_{3+3i})=[3]$ for $0 \leq i \leq \frac{n-3}{3}$ and $h(x)= \emptyset$ otherwise.
It is easy to see that $h$ is a M3RDF of $C_n$ with weight $\frac{4n}{3}$.
Thus, we have $\gamma_{r3}^\star(C_n) = \frac{4n}{3}$.

\end{proof}

\begin{prop}\label{prop:Kn}
For $n \geq 2$,
\begin{equation*}
\gamma_{r3}^\star(K_n) = \left\{
                      \begin{array}{ll}
                     \frac{3n}{2} & \hbox{if $n$ is even;} \\
                     \frac{3n -1}{2} & \hbox{if $n$ is odd.}
                      \end{array}
                     \right.
\end{equation*}
\end{prop}
\begin{proof}
Let $V(K_n) =\{v_1, \dotsc, v_n\}$.
If $n$ is odd, then define $g : V(K_n) \cup E(K_n) \rightarrow \mathcal{P}([3])$ by $g(v_1)=\{1\}$, $g(v_{2i}v_{2i+1})=[3]$
for $1 \leq i \leq \frac{n-1}{2}$ and $g(x)=\emptyset$ otherwise.
It is easy to see that $g$ is a M3RDF of $K_n$ with weight $\frac{3n-1}{2}$.
If $n$ is even, then define $h : V(K_n) \cup E(K_n) \rightarrow \mathcal{P}([3])$ by $h(v_{2i-1}v_{2i})=[3]$
for $1 \leq i \leq \frac{n}{2}$ and $h(x)=\emptyset$ otherwise.
It is easy to see that $h$ is a M3RDF of $K_n$ with weight $\frac{3n}{2}$.

Now we claim that $\gamma_{r3}^\star(K_n) \geq \frac{3n}{2}$ if $n$ is even and $\gamma_{r3}^\star(K_n) \geq \frac{3n-1}{2}$ if $n$ is odd.
One can easily check that $\gamma_{r3}^\star(K_2)=3$, $\gamma_{r3}^\star(K_3)=4$,
and $\gamma_{r3}^\star(K_4)= 6$.
We proceed by induction on $n$.
Assume that $n \geq 5$.

\textbf{Case 1.} $n$ is odd.

Let $f$ be a $\gamma_{r3}^\star(K_n)$-function.
If $\sum_{i=1}^n|f(v_i)| =0$, then to dominate each $v_i$ we must have $[3] \subseteq \cup_{e \in N_M(v_i)}f(e)$ for each $v_i \in V(K_n)$.
This implies that
\[2\gamma_{r3}^\star(K_n) = \sum_{i=1}^n \sum_{e \in N_M(v_i)} |f(e)| \geq 3n.\]
Thus, $\gamma_{r3}^\star(K_n) > \frac{3n-1}{2}$.

Assume that $\sum_{i=1}^n|f(v_i)| \neq 0$.
Without loss of generality, assume that $f(v_1) \neq \emptyset$.
Define $g : V(K_n -v_1) \cup E(K_n -v_1) \rightarrow \mathcal{P}([3])$ by $g(v_i) = f(v_i) \cup f(v_1v_i)$
for $2 \leq i \leq n$ and $g(x)=f(x)$ otherwise.
Then $g$ is a M3RDF with weight at most $\gamma_{r3}^\star(K_n) -1$.
By the induction  hypothesis, we have
\[\gamma_{r3}^\star(K_n) \geq \omega(g) +1 \geq \gamma_{r3}^\star(K_{n-1}) +1 = \frac{3(n-1)}{2} +1 =\frac{3n -1}{2}.\]

\vskip5pt
\textbf{Case 2.} $n$ is even.

We choose a $\gamma_{r3}^\star(K_n)$-function $f$ so that the size of $\{v_i \in V(K_n) \mid f(v_i) \neq \emptyset\}$ is as small as possible.

For an edge $v_iv_j \in E(K_n)$, if $|f(v_i)| + |f(v_j)| + |f(v_iv_j)| = 3$, then
define $h : V(K_n -\{v_i, v_j\}) \cup E(K_n -\{v_i, v_j\}) \rightarrow \mathcal{P}([3])$ by $h(v_k)= f(v_k) \cup f(v_kv_i) \cup f(v_kv_j)$
for $v_k \in V(K_n) \setminus \{v_i, v_j\}$ and $h(x)=f(x)$ otherwise.
Clearly $h$ is a M3RDF with weight $\gamma_{r3}^\star(K_n) -3$.
By the induction  hypothesis, we have
\[\gamma_{r3}^\star(K_n) \geq \omega(h) +3 \geq \gamma_{r3}^\star(K_{n-2}) +3 = \frac{3(n-2)}{2} +3 =\frac{3n}{2}.\]

Assume that $|f(v_i)| + |f(v_j)| + |f(v_iv_j)| \leq 2$ for any $v_i, v_j \in V(K_n)$.

For a vertex $v_i \in V(K_n)$, if $|f(v_i)|=2$, then
define $g : V(K_n - v_i) \cup E(K_n - v_i) \rightarrow \mathcal{P}([3])$ by $g(v_k)= f(v_k) \cup f(v_kv_i)$
for $v_k \in V(K_n) \setminus \{v_i\}$ and $g(x)=f(x)$ otherwise.
Clearly $g$ is a M3RDF with weight $\gamma_{r3}^\star(K_n) -2$.
By the induction  hypothesis, we have
\[\gamma_{r3}^\star(K_n) \geq \omega(g) +2 \geq \gamma_{r3}^\star(K_{n-1}) +2   = \frac{3(n-1)-1}{2} +2 =\frac{3n}{2}.\]

Assume that $|f(v_i)| \leq 1$ for each $v_i \in V(K_n)$.

Let $N:= \{ v_i \in V(K_n) \mid f(v_i) \neq \emptyset \}$.
For a fixed $v_i \in N$, suppose that there exists no $v_j \in V(K_n)$ such that $f(v_iv_j) \neq \emptyset$.
Without loss of generality, we may assume that $f(v_i) = \{ 1 \}$.
To dominate elements in $\{v_iv_k \mid v_k \in V(K_n) \setminus \{v_i\} \}$,
$\bigcup_{x \in N_M[v_k]} f(x)$ should contain $2$ and $3$ for each $v_k \in V(K_n) \setminus \{v_i\}$.
Thus, we have $s_2, s_3 \geq \frac{n}{2}$,
where $s_j := |\{x \in V(K_n - v_i) \cup E(K_n - v_i) \mid j \in f(x) \}|$.
If $f(x) \neq \emptyset$ for all $x \in N_M[v_k] \setminus \{v_iv_k\}$, then $\sum_{x \in N_M[v_k]} |f(x)| \geq n$.
This implies $\gamma_{r3}^\star(K_n) \geq \frac{3n}{2}$.
Assume that for $v_k \in V(K_n) \setminus \{v_i\}$ there exists $x \in N_M[v_k] \setminus \{v_iv_k\}$ such that $f(x)=\emptyset$.
Then to dominate such an element $x$, $\bigcup_{y \in N_M(x)} f(y)$ should contain $1$.
Thus, we have $s_1 \geq \frac{n}{2}$ so that $\omega(f) \geq |f(v_i)| + s_2 + s_3 + s_1 = 1+ \frac{3n}{2}$.

Assume that for each $v_i \in N$ there exists a vertex $v_j \in V(K_n)$ such that $f(v_iv_j) \neq \emptyset$.
Let $t$ be the size of set $N$.
For $v_j \in V(K_n) \setminus N$, $\bigcup_{x \in N_M[v_j]} f(x)$ should contain $[3]$.
For $v_i \in N$, $\bigcup_{x \in N_M(v_i)} f(x)$ contains at least one element.
Thus, we have $\gamma_{r3}^\star(K_n) \geq \frac{3(n-t)}{2} + (t + \frac{t}{2}) = \frac{3n}{2}$.
This completes the proof.
\end{proof}

\section{Lower and upper bounds for trees}\label{sec:main3}

In this section, we provide lower and upper bounds for the middle $3$-rainbow domination number of trees in terms of the matching number.

\begin{thm}\label{thm:simple-upper}
For every tree $T$ of order $n$, $\gamma_{r3}^\star(T)  \leq  n +\alpha'(T)$.
\end{thm}
\begin{proof}
Take a maximum matching $M$ in $T$.
Let $U$ be the set of vertices which are not saturated by $M$.
Define a function by $f(e)=[3]$ for $e \in M$, $f(u)=\{1\}$ for $u \in U$ and $f(x)=\emptyset$ otherwise.
Clearly, $f$ is a M3RDF of $T$. Thus, $\gamma_{r3}^\star(T) \leq |U| +3|M| = (n- 2\alpha'(T)) + 3\alpha'(T) = n + \alpha'(T)$.
\end{proof}

\begin{lem}\label{lem:sum}
Let $T$ be a tree and $P_3=uvw$ a path in $T$ with $deg_T(v)=2$ and $deg_T(w)=1$.
Then $|f(uv)| + |f(v)| + |f(vw)| + |f(w)| \geq 3$ for any $\gamma_{r3}^\star(T)$-function $f$.
\end{lem}
\begin{proof}
If $f(w) = \emptyset$, then to dominate $w$ it follows that $f(vw)=[3]$.
If $f(vw) = \emptyset$, then to dominate $vw$ it follows that $f(uv) \cup f(v) \cup  f(w) = [3]$.
If $f(v) = \emptyset$, then to dominate $v$ it follows that $f(uv) \cup  f(vw) = [3]$.
In the above cases, we have $|f(uv)| + |f(v)| + |f(vw)| + |f(w)| \geq 3$.
If  $f(w)$, $f(vw)$ and $f(v)$ are not empty, then $|f(uv)| + |f(v)| + |f(vw)| + |f(w)| \geq 3$.
This completes the proof.
\end{proof}

\begin{thm}\label{thm:trees}
For every tree $T$, $\gamma_{r3}^\star(T) \geq \frac{5\alpha'(T)}{2}$.
\end{thm}
\begin{proof}
We proceed by induction on the order $n$ of $T$.
Obviously, the statement is true for all trees of order $n \leq 4$.

Let $T$ be a tree of order $n \geq 5$.
Suppose that every tree $T'$ of order $n' (< n)$ satisfies $\gamma_{r3}^\star(T') \geq  \frac{5\alpha'(T')}{2}$.
Let $M$ be a maximum matching in $T$.
If $T$ is a star, then $\alpha'(T)=1$ and so $n+1 = \gamma_{r3}^\star(T) > \frac{5}{2}$.
Assume that $T$ is a double star $DS_{p,q}$ with $p \geq q \geq 1$.
Then $\alpha'(T)=2$ and so $n+1 = \gamma_{r3}^\star(T) > 5$.
Now we assume that $T$ is neither a star or a double star.
Then it is easy to see that $T$ has diameter at least four.

If $T$ has a pendant edge $uv$ such that $v$ is a leaf and $uv \not\in M$, then
for any $\gamma_{r3}^\star(T)$-function $f$,
the function $g : V(T -v) \cup E(T -v)$ defined by by $g(u)=f(u) \cup f(uv)$ and $g(x)=f(x)$ otherwise is a M3RDF of $T -v$ with weight at most $\omega(f)$.
By the induction hypothesis, we have $\gamma_{r3}^\star(T) \geq \gamma_{r3}^\star(T -v) \geq \frac{5\alpha'(T-v)}{2} = \frac{5\alpha'(T)}{2}$.
Thus, we assume the following.

\vskip5pt
\textbf{Assumption 1.} All pendant edges of $T$ belong to each maximum matching.

Then it follows that all support vertices have degree $2$.
If $\Delta(T)=2$, then $T$ is a path. So, the result follows by Proposition \ref{prop:path}.
From now on, assume that $\Delta(T) \geq 3$.
Among all of diametrical paths in $T$, we choose $x_0x_1\dotsc x_d$ so that it maximizes the size of $f(x_{d-2})$.
Root $T$ at $x_0$.
It follows from Lemma \ref{lem:sum} that $|f(x_{d-2}x_{d-1})| + |f(x_{d-1})| + |f(x_{d-1}x_{d})| + |f(x_{d})| \geq 3$.
We divide our consideration into three cases.

\vskip5pt
\textbf{Case 1.} $deg(x_{d-2}) \geq 3$.

First, suppose that there is a path $x_{d-2}yz$ in $T$ such that $z$ is a leaf and $y \not\in \{x_{d-3}, x_{d-1}\}$.
Then it follows from Lemma \ref{lem:sum} that $|f(z)| + |f(yz)| + |f(y)| + |f(x_{d-2}y)| \geq 3$.
Without loss of generality, we may assume that
\begin{equation}\label{match1}
|f(x_{d-2}y)| \geq |f(x_{d-2}x_{d-1})|.
\end{equation}

Let $T' = T - \{ x_{d-1}, x_d\}$.
Then clearly $\alpha'(T') = \alpha'(T) -1$.
Since both $f(x_{d-2}y)$ and $f(x_{d-2}x_{d-1})$ are not $[3]$,
the assumption (\ref{match1}) implies that $f|_{V(T') \cup E(T')}$ is a M3RDF of $T'$ with weight at most $\omega(f)-3$.
By the induction hypothesis, we have
$\gamma_{r3}^\star(T) \geq \gamma_{r3}^\star(T') +3 \geq \frac{5\alpha'(T')}{2} +3 = \frac{5(\alpha'(T)-1)}{2} +3 > \frac{5\alpha'(T)}{2}$.

Now assume that every element of $C(x_{d-2}) \setminus \{ x_{d-1}\}$ is a leaf. Then it follows from Assumption 1
that $|C(x_{d-2}) \setminus \{ x_{d-1}\}|=1$.
Let $v \in C(x_{d-2}) \setminus \{ x_{d-1}\}$.
If $f(v)=\emptyset$, then $x_{d-2}v$ can not dominate $v$. Thus, we must have $f(v) \neq \emptyset$.
We consider the following subcases.

\vskip5pt
\textbf{Subcase 1.1.}  $|f(x_{d-2}x_{d-1})| \leq 1$.

This implies that $|f(x_{d-1})| + |f(x_{d-1}x_{d})| + |f(x_{d})| = 3$.
Let $T' = T - \{ x_{d-1}, x_d\}$.
Define $g : V(T') \cup E(T') \rightarrow \mathcal{P}([3])$ by $g(x_{d-2})=f(x_{d-2}) \cup f(x_{d-2}x_{d-1})$ and $g(x)=f(x)$ otherwise.
Clearly $g$ is a M3RDF of $T'$ with weight $\omega(f)-3$.
The result follows as above.

\vskip5pt
\textbf{Subcase 1.2.} $|f(x_{d-2}x_{d-1})| \geq 2$.

It is easy to see that $|f(x_{d-1})| + |f(x_{d-1}x_d)| + |f(x_d)| = 4 - |f(x_{d-2}x_{d-1})|$.
Let $T'= T - \{x_{d-2}, v, x_{d-1}, x_d\}$.
Define $g : V(T') \cup E(T') \rightarrow \mathcal{P}([3])$ by $g(x_{d-3})=f(x_{d-3}) \cup f(x_{d-3}x_{d-2})$ and $g(x)=f(x)$ otherwise.
Clearly $g$ is a M3RDF of $T'$ with weight $\omega(f)-5$.
By the induction hypothesis, we have
$\gamma_{r3}^\star(T) \geq \gamma_{r3}^\star(T') +5 \geq \frac{5\alpha'(T')}{2} +5 = \frac{5(\alpha'(T)-2)}{2} +5 = \frac{5\alpha'(T)}{2}$.

\vskip5pt
\textbf{Case 2.} $deg(x_{d-2}) = 2$ and $deg(x_{d-3}) \geq 3$.

It follows from Assumption 1 that $x_{d-3}x_{d-2}, x_{d-1}x_d \in M$.
If there exists a path $x_{d-3}xyz$ in $T$ such that $x \not\in \{x_{d-2}, x_{d-4} \}$ and $z$ is a leaf.
By Case 1 and Assumption 1, we may assume that $deg(x)=deg(y)=2$.
It follows from Assumption 1 and $x_{d-3}x_{d-2} \in M$ that $x_{d-3}x ,xy \not\in M$.
But, $(M \setminus \{yz\}) \cup \{xy\}$ is a maximum matching in $T$ not containing a pendant edge, a contradiction.

It follows from $deg(x_{d-3}) \geq 3$ that there exists a path $x_{d-3}yz$ such that $deg(y)=2$ and $z$ is a leaf.
By Lemma \ref{lem:sum}, we have
$|f(x_{d-2}x_{d-1})| + |f(x_{d-1})| + |f(x_{d-1}x_d)| + |f(x_d)| \geq 3$ and
$|f(x_{d-3}y)| + |f(y)| + |f(yz)| + |f(z)| \geq 3$.

Since $f$ is a $\gamma_{r3}^\star(T)$-function, it is easy to see that
\[4 \geq S(x_{d-3}yz) \geq 3 ~\text{and}~ 6 \geq S(x_{d-3}x_{d-2}x_{d-1}x_d) \geq 4,\]
where
$S(x_{d-3}yz) := |f(x_{d-3})| + |f(x_{d-3}y)| + |f(y)| + |f(yz)| + |f(z)|$ and
$S(x_{d-3}x_{d-2}x_{d-1}x_d):= |f(x_{d-3})|+ |f(x_{d-3}x_{d-2})| + |f(x_{d-2})| + |f(x_{d-2}x_{d-1})| + |f(x_{d-1})| + |f(x_{d-1}x_d)| + |f(x_d)|$.

If $S(x_{d-3}x_{d-2}x_{d-1}x_d) =6$, then without loss of generality we may assume that $f(x_{d-3}x_{d-2})= f(x_{d-1}x_d) =[3]$.
Let $T'= T -\{x_{d-1}, x_d\}$.
The function $f|_{V(T') \cup E(T')}$ is a M3RDF of $T'$ with weight $\omega(f)-3$.
By the induction hypothesis, we have
$\gamma_{r3}^\star(T) \geq \gamma_{r3}^\star(T') +3 \geq \frac{5\alpha'(T')}{2} +3 = \frac{5(\alpha'(T)-1)}{2} +3 > \frac{5\alpha'(T)}{2}$.

If $S(x_{d-3}yz) =3$, then $|f(y)| + |f(yz)| + |f(z)|=3$ and $f(x_{d-3}) = f(x_{d-3}y) = \emptyset$.
Let $T'= T -\{y, z\}$.
The function $f|_{V(T') \cup E(T')}$ is a M3RDF of $T'$ with weight $\omega(f)-3$.
The result follows as above.

Now we assume that $S(x_{d-3}yz) = 4$ and $5 \geq S(x_{d-3}x_{d-2}x_{d-1}x_d) \geq 4$.
Then without loss of generality we may assume that $f(x_{d-3}y)=[3]$ and $f(z)=\{1\}$.

If $S(x_{d-3}x_{d-2}x_{d-1}x_d)=5$, then define $g : V(T) \cup E(T) \rightarrow \mathcal{P}([3])$ by $g(x_{d-2})=\{1\}$,
$g(x_{d-1}x_d)=[3]$ and $f(x_{d-3})= f(x_{d-3}x_{d-2})= f(x_{d-2}x_{d-1}) =f(x_{d-1}) = f(x_d) = \emptyset$.
Clearly $g$ is a M3RDF of $T$ with weight $\omega(f) -1$, a contradiction.

Assume that $S(x_{d-3}x_{d-2}x_{d-1}x_d)=4$.
Then without loss of generality we may assume that $f(x_{d-2})=\{1\}$ and $f(x_{d-1}x_d)=[3]$.
Let $T' = T - \{x_{d-1}, x_d\}$.
The function $f|_{V(T') \cup E(T')}$ is a M3RDF of $T'$ with weight $\omega(f)-3$.
The result follows as above.

\vskip5pt
\textbf{Case 3.} $deg(x_{d-2}) = 2$ and $deg(x_{d-3}) =2$.

If $f(x_{d-4}x_{d-3}) =[3]$, then it is easy to see that
$|f(x_{d-3})| + |f(x_{d-3}x_{d-2})| + |f(x_{d-2})| + |f(x_{d-2}x_{d-1})| + |f(x_{d-1})|+ |f(x_{d-1}x_d)| + |f(x_d)| = 4$,
since $f$ is a $\gamma_{r3}^\star(T)$-function.
Without loss of generality,
we may assume that $f(x_{d-1}x_d)=[3]$, $f(x_{d-2})= \{1\}$ and $f(x_{d-3})= f(x_{d-3}x_{d-2}) =f(x_{d-2}x_{d-1})=f(x_{d-1})=f(x_d) = \emptyset$.
Let $T' = T - \{ x_{d-1}, x_d\}$.
Then clearly $\alpha'(T') = \alpha'(T) -1$ and $f|_{V(T') \cup E(T')}$ is a M3RDF of $T'$ with weight at most $\omega(f)-3$.
By the induction hypothesis, we have
$\gamma_{r3}^\star(T) \geq \gamma_{r3}^\star(T') +3 \geq \frac{5\alpha'(T')}{2} +3 = \frac{5(\alpha'(T)-1)}{2} +3 > \frac{5\alpha'(T)}{2}$.

If $f(x_{d-4}x_{d-3}) \neq [3]$, then it is easy to see that
$|f(x_{d-3})| + |f(x_{d-3}x_{d-2})| + |f(x_{d-2})| + |f(x_{d-2}x_{d-1})| + |f(x_{d-1})|+ |f(x_{d-1}x_d)| + |f(x_d)| = 5$,
since $f$ is a $\gamma_{r3}^\star(T)$-function.
Without loss of generality, we may assume that $f(x_{d-3})= f(x_d) =\{1\}$, $f(x_{d-2}x_{d-1})=[3]$ and $f(x_{d-3}x_{d-2}) = f(x_{d-2}) = f(x_{d-1}) = f(x_{d-1}x_d)=\emptyset$.

Let $T' = T - T_{x_{d-3}}$.
Define $g : V(T') \cup E(T') \rightarrow \mathcal{P}([3])$ by $g(x_{d-4})= f(x_{d-4}) \cup f(x_{d-4}x_{d-3})$ and $g(x)=f(x)$ otherwise.
Then clearly $\alpha'(T') = \alpha'(T) -2$ and $g$ is a M3RDF of $T'$ with weight at most $\omega(f)-5$.
By the induction hypothesis, we have
$\gamma_{r3}^\star(T) \geq \gamma_{r3}^\star(T') +5 \geq \frac{5\alpha'(T')}{2} +5 = \frac{5(\alpha'(T)-2)}{2} +5 = \frac{5\alpha'(T)}{2}$.
This completes the proof.
\end{proof}

\section{The $3$-rainbow domatic number for the middle graph of paths and cycles}\label{sec:main4}

In this section, we determine the $3$-rainbow domatic number for the middle graph of paths and cycles.

\begin{thm}[See \cite{SV}]\label{upb}
If $G$ is a graph of order $n$, then $\gamma_{rk}(G) \cdot d_{rk}(G) \leq kn$.
\end{thm}

\begin{thm}[See \cite{SV}]\label{upb2}
For every graph $G$, $d_{rk}(G) \leq \delta(G) +k$.
\end{thm}

\begin{prop}\label{dr3:path-even}
For $n \geq 4$ and $n \equiv 0$ (mod $2$),
$d_{r3}(M(P_n)) = 4$.
\end{prop}
\begin{proof}
By Theorem \ref{upb2}, $d_{r3}(M(P_n)) \leq 4$.

Let $P_n = v_1v_2\dotsc v_n$.
Define the $3$-rainbow dominating functions $f_1, f_2, f_3, f_4$ as follows:

$f_1(v_{1+2i}v_{2+2i})=[3]$ for $0 \leq i \leq \frac{n-2}{2}$ and $f_1(x)=\emptyset$ otherwise,

$f_2(v_{1+2i})=\{1\}$, $f_2(v_{2+2i})=\{2\}$ for $0 \leq i \leq \frac{n-2}{2}$, $f_2(v_{2+2i}v_{3+2i})=\{3\}$ for $0 \leq i \leq \frac{n-4}{2}$ and $f_2(x)=\emptyset$ otherwise,

$f_3(v_{1+2i})=\{2\}$, $f_3(v_{2+2i})=\{3\}$ for $0 \leq i \leq \frac{n-2}{2}$, $f_3(v_{2+2i}v_{3+2i})=\{1\}$ for $0 \leq i \leq \frac{n-4}{2}$ and $f_3(x)=\emptyset$ otherwise,

$f_4(v_{1+2i})=\{3\}$, $f_4(v_{2+2i})=\{1\}$ for $0 \leq i \leq \frac{n-2}{2}$, $f_4(v_{2+2i}v_{3+2i})=\{2\}$ for $0 \leq i \leq \frac{n-4}{2}$ and $f_4(x)=\emptyset$ otherwise.

Then clearly $f_i$ is a $3$-rainbow dominating function on $M(P_n)$ for each $i$.
Thus, $\{f_1, f_2, f_3, f_4\}$ is a family of $3$-rainbow dominating functions on $M(P_n)$.

\end{proof}

%Question:

%For $n \geq 5$ and $n \equiv 1$ (mod $2$),
%$d_{r3}(M(P_n)) = 3$. ?

%\begin{prop}\label{dr3:cycle}
%For $n \geq 3$ and $n \equiv 0$ (mod $3$),
%$d_{r3}(M(C_n)) = 4$.
%\end{prop}
%\begin{proof}
%By Theorem \ref{upb}, $\gamma_{r3}(M(C_n)) \cdot d_{r3}(M(C_n)) \leq 3\cdot2n$.
%It follows from Proposition \ref{prop:cycle} that $d_{r3}(M(C_n)) \leq 4$.

%Let $C_n = v_1v_2\dotsc v_nv_1$.
%Define the $3$-rainbow dominating functions $f_1, f_2, f_3, f_4$ as follows:

%$f_1(v_{1+3i}v_{2+3i})=[3]$, $f_1(v_{3+3i})=\{1\}$ for $0 \leq i \leq \frac{n-3}{3}$ and $f_1(x)=\emptyset$ otherwise,

%$f_2(v_{2+3i}v_{3+3i})=[3]$, $f_2(v_{4+3i})=\{1\}$ for $0 \leq i \leq \frac{n-3}{3}$ and $f_2(x)=\emptyset$ otherwise,

%$f_3(v_{3+3i}v_{4+3i})=[3]$, $f_3(v_{5+3i})=\{1\}$ for $0 \leq i \leq \frac{n-3}{3}$ and $f_3(x)=\emptyset$ otherwise,

%$f_4(v_{1+3i})= \{1,2\}$, $f_4(v_{2+3i})=\{2,3\}$, $f_4(v_{3+3i})=\{3,1\}$ for $0 \leq i \leq \frac{n-3}{3}$ and $f_4(x)=\emptyset$ otherwise.

%Then clearly $f_i$ is a $3$-rainbow dominating function on $M(C_n)$ for each $i$.
%Thus, $\{f_1, f_2, f_3, f_4\}$ is a family of $3$-rainbow dominating functions on $M(C_n)$.
%\end{proof}

\begin{prop}\label{dr3:cycle1}
For $n \geq 4$,
$d_{r3}(M(C_n)) = 4$.
\end{prop}
\begin{proof}
By Theorem \ref{upb}, $\gamma_{r3}(M(C_n)) \cdot d_{r3}(M(C_n)) \leq 3\cdot2n$.
It follows from Proposition \ref{prop:cycle} that $d_{r3}(M(C_n)) \leq 4$.

Let $C_n = v_1v_2\dotsc v_nv_1$.
We consider the following two cases.

\vskip5pt
\textbf{Case 1.} $n$ is even.

Extend the $3$-rainbow dominating functions $f_1, f_2, f_3, f_4$ in Proposition \ref{dr3:path-even} as follows:

$g_1(v_nv_1)=\emptyset$ and $g_1(x) =f_1(x)$ otherwise,

$g_2(v_nv_1)=\{1\}$ and $g_2(x) =f_2(x)$ otherwise,

$g_3(v_nv_1)=\{1\}$ and $g_3(x) =f_3(x)$ otherwise,

$g_4(v_nv_1)=\{1\}$ and $g_4(x) =f_4(x)$ otherwise.

\vskip5pt
\textbf{Case 2.}  $n$ is odd.

Then $n-1 \geq 4$ is even.
Let $f_1, f_2, f_3, f_4$ be the $3$-rainbow dominating functions on $M(P_{n-1})$ given by Proposition \ref{dr3:path-even}.
Extend them as follows:

$g_1(v_{n-1}v_n)=\emptyset$, $g_1(v_n)=\{1\}$, $g_1(v_nv_1) =\emptyset$ and $g_1(x) =f_1(x)$ otherwise,

$g_2(v_{n-1}v_n)=\{3\}$, $g_2(v_n)=\{2\}$, $g_2(v_nv_1) =\emptyset$ and $g_2(x) =f_2(x)$ otherwise,

$g_3(v_{n-1}v_n)=\{1\}$, $g_3(v_n)=\{3\}$, $g_3(v_nv_1) =\emptyset$ and $g_3(x) =f_3(x)$ otherwise,

$g_4(v_{n-1}v_n)=\emptyset$, $g_4(v_n)=\emptyset$, $g_4(v_nv_1) =[3]$ and $g_4(x) =f_4(x)$ otherwise.

\vskip5pt
In any case, clearly $g_i$ is a $3$-rainbow dominating function on $M(C_n)$ for each $i$.
Thus, $\{g_1, g_2, g_3, g_4\}$ is a family of $3$-rainbow dominating functions on $M(C_n)$.
\end{proof}

\bibstyle{plain}

\end{document}